\documentclass[onecolumn,11pt,aps,nofootinbib,superscriptaddress,tightenlines,showkeys]{revtex4}

\usepackage[T1]{fontenc}
\usepackage[utf8]{inputenc}

\usepackage[cmex10]{amsmath}
\usepackage{graphics}
\usepackage{dsfont}
\usepackage{amssymb}
\usepackage{graphicx}
\usepackage{mathrsfs}
\usepackage{units}
\usepackage{pgfplots}
\usepackage{enumitem}
\usepackage[colorlinks=true,linkcolor=black,citecolor=black,plainpages=false,pdfpagelabels]{hyperref}
\hypersetup{pdftitle={Eigenvalue estimates for the resolvent of a non-normal matrix}}

\usepackage{amsthm}
\newtheorem{theorem}{Theorem}[section]

\newtheorem{proposition}[theorem]{Proposition}

\newtheorem{definition}{Definition}[section]

\newenvironment{remark}[1][Remark]{\begin{trivlist}
\item[\hskip \labelsep {\bfseries #1}]}{\end{trivlist}}
\newenvironment{remarks}[1][Remarks]{\begin{trivlist}
\item[\hskip \labelsep {\bfseries #1}]}{\end{trivlist}}

\newcommand{\id}{\ensuremath{\mathds{1}}}

\begin{document}
\title{Optimal Injectivity Conditions for Bilinear Inverse Problems with Applications to Identifiability of Deconvolution Problems}
\author{Michael Kech}
\email{kech@ma.tum.de}
\affiliation{Department of Mathematics, Technische Universit\"{a}t M\"{u}nchen, 85748 Garching, Germany}
\author{Felix Krahmer}
\email{felix.krahmer@tum.de}
\affiliation{Department of Mathematics, Technische Universit\"{a}t M\"{u}nchen, 85748 Garching, Germany}

\date{\today}

\begin{abstract}
We study identifiability for bilinear inverse problems under sparsity and subspace constraints. We show that, up to a global scaling ambiguity, almost all such maps are injective on the set of pairs of sparse vectors if the number of measurements $m$ exceeds $2(s_1+s_2)-2$, where $s_1$ and $s_2$ denote the sparsity of the two input vectors, and injective on the set of pairs of vectors lying in known subspaces of dimensions $n_1$ and $n_2$ if $m\geq 2(n_1+n_2)-4$. We also prove that both these bounds are tight in the sense that one cannot have injectivity for a smaller number of measurements. Our proof technique draws from algebraic geometry. As an application we derive optimal identifiability conditions for the deconvolution problem, thus improving on recent work of Li et al.~\cite{li2015identifiability}.

\end{abstract}

\maketitle

\tableofcontents

\section{Introduction}
While inverse problems have been subject of study for many years, a new viewpoint has been taken in the last years, starting with  the fundamental works on compressed sensing \cite{CRT06,do06b}. Namely, rather than assuming that the inverse problem is completely determined by the application and potentially maximally ill-posed, the paradigm was to use the remaining degrees of freedom in the problem design as much as possible to ensure a unique solution. 

A common observation in many scenarios is that most measurement setups behave near optimally. There are two ways to make this precise and works following both approaches: On the one hand, one can choose the measurement parameters at random and study conditions entailing that recovery is possibly with high probability. On the other hand one can take an information theoretic approach, aiming to establish identifability (that is, solution uniqueness or injectivity) for all measurement parameters except for a set of measure zero.
In contrast to the first approach, the goal is not to devise working solutions, but rather to establish fundamental limits for the number of measurements, which can then be used as a measure to judge the quality of a concrete measurement setup.
As a consequence, the number of measurements needed in the second approach is smaller, but the resulting statement is weaker in the sense that much weaker claims are made about the stability with respect to noise and none about algorithms to find the unique solution.

These approaches have been extensively applied to linear inverse problems, that is, one considers linear measurements of a signal known to satisfy some additional model assumptions.
In this context, the randomized approach is addressed in many works in the areas of compressed sensing -- here the signal is assumed to be sparse --, low rank matrix recovery, and beyond (see \cite{foucart_mathematical_2013} for a textbook with many references). The article \cite{chparewi10} provides a comprehensive treatment to randomized linear inverse problems on a very general level via convex optimization. 
The linear setup has also been studied from the identifiability viewpoint. Under sparsity assumptions, this problem relates to the study of the spark of the measurement matrix, that is, (one more than) the largest number of linearly independent columns (cf.~\cite{ACM12}). For low rank matrix recovery, identifiability conditions have been established in \cite{ENP12}. Again for low rank matrix recovery, identifiability conditions for random signal models have recently been studied in \cite{RSB15}.

More recently, similar considerations have also been applied to the phase retrieval problem, where one still considers linear measurements, but only the (square of) their absolute values is observed. Recovery guarantees for randomized setups have first been proven in \cite{CSV13} and extended in many follow-up works. At the core of many of these works is a lifting idea, namely that phaseless measurements can be expressed as linear measurements on the outer product of the signal with itself. For the study of injectivity conditions, an algebraic geometry viewpoint has proven useful. It turned out to be crucial whether one deals with real or complex measurements. Optimal injectivity bounds (up to a global phase factor) for the real case were proven in \cite{balan_signal_2006}, while the complex case has proven significantly more difficult and bounds have successively improved over the last few years \cite{balan_signal_2006, CDHV15, V15}.

For bilinear inverse problems, i.e., measurements that depend in a bilinear way on two input signals, such considerations are only in their beginnings. Here one can only hope for injectivity up to a global multiplicative constant. Again, via a lifting approach, such problems can be identified with low rank matrix recovery problems. So if no structural constraints are imposed, recovery guarantees directly carry over. 
The setup with additional sparsity assumptions, again under random measurements, is somewhat more involved and has been studied in \cite{LWB13}. Deterministic conditions for injectivity are derived in \cite{LLB15}, where the authors also consider scenarios where one only finds uniqueness up to more general multiplication groups.

Without additional assumptions, there is still a gap between both the randomized and the algebraic setup and measurement systems arising in applications. Namely, applications impose additional constraints on the structure of the measurements, which correspond to a measure zero set in the space of unrestricted measurements, so the aforementioned results typically do not have implications about whether any of the solutions satisfy these constraints. This motivated the study of structured random measurements in all the scenarios mentioned (see \cite{KR14} for a survey on such approaches, not yet including bilinear problems). In compressed sensing, for example, randomly subsampled Fourier measurements have been considered \cite{ruve08}, as motivated by applications in magnetic resonance imaging \cite{ldp07}, as well as subsampled convolutions with a random vector \cite{KMR12}, as motivated by applications in remote sensing and coded aperture imaging \cite{mawi08}. For low-rank matrix recovery, a measurement model arising in collaborative filtering application consists of randomly selected of matrix entries, which yields the so-called matrix completion problem \cite{care09}. For the phase retrieval problem, concatenations of Fourier measurements and random diagonal matrices have been studied, which are motivated by the idea of introducing a mask in a diffraction imaging setup \cite{BCM14, CLS15, GKK15}.

Lastly, for bilinear inverse problems, two important classes of models that have been studied are calibration problems as well as blind deconvolution and demixing problems. From a randomization viewpoint, calibration problems have been studied in \cite{LS15a}, blind deconvolution problems are studied in \cite{ARR14}, and blind demixing problems are studied in \cite{LS15b}. All these papers are based again on lifting ideas. Identifiability conditions for calibration problems are derived in \cite{LLB15}.

Identifiability for the blind deconvolution problem under sparsity or subspace constraints has first been studied in \cite{CM14, CM14b}, in particular providing negative results for signals sparse in the standard basis. Subsequently, this case has been identified as exceptional by providing identifiability results that hold for all sparsity bases  except for a set of measure zero \cite{LLB15a}. These results have then been improved to a near-optimal number of measurements \cite{li2015identifiability}. The authors distinguish between weak and strong identifiability. The former notion relates to the number of measurements  needed to ensure that for a given fixed signal and a generic set of sparsity bases, there is no other signal resulting in the same measurements; the latter requires that property uniformly for all potential signal. In the case of weak identifiability, the set of measurements, where the property fails is allowed to differ for each signal, so there will not necessarily be a set of bases for which one has injectivity, i.e., uniqueness for all signals at the same time.

Conditions for strong identifibiability in the blind deconvolution problem are also the main application of our results. The same techniques also yield conditions for weak identifiability (see Appendix~\ref{app}). In contrast to the result in \cite{li2015identifiability}, the our identifiability results for both cases are tight, that is, our theory implies matching upper and lower bounds for the number of measurements needed. Also, our proof techniques are very different to those in \cite{li2015identifiability}; our work is mainly based on techniques from algebraic geometry.

\textit{Outline.} In Section \ref{0} we fix notation and review the notions of weak and strong identifiability for bilinear inverse problems. 

Then, in Section \ref{main}, we give the main results of the present paper. We aim at discriminating any 
two pairs of vectors $(v,w),(v^\prime,w^\prime)\in \mathbb{C}^{n_1}\times \mathbb{C}^{n_2}$ up to the trivial scaling ambiguity from the outcomes $B(v,w)$, $B(v^\prime,w^\prime)$ of a bilinear measurement map $B:\mathbb{C}^{n_1}\times\mathbb{C}^{n_2}\to\mathbb{C}^m$ under the premise that $v,v^\prime$ are $s_1$-sparse and $w,w^\prime$ are $s_2$-sparse. We show that a bilinear map performing this task exists if and only if $m\geq 2(n_1+n_2)-4$ in case $s_1=n_1$, $s_2=n_2$ and $m\geq 2(s_1+s_2)-2$ otherwise.

In the second part of this section we apply our results to derive strong identfiability conditions for the deconvolution problem, which aims at identifying a signal $v\in\mathbb{C}^m$ and a filter $w\in\mathbb{C}^m$ up to the trivial scaling ambiguity from their circular convolution $v\circledast w$, and compare our results to previous work. Dimension counting already implies that blind deconvolution is infeasible in general, however identifiability may be possible when assuming $v\in V$ and $w\in W$ for some lower dimensional subspaces $V,W\subseteq\mathbb{C}^m$. Indeed we show that for generic subspaces $V,W\in\mathbb{C}^m$ it is possible to discriminate any two pairs of vectors $(v,w),(v^\prime,w^\prime)\in V\times W$ up to the trivial scaling ambiguity from their circular convolutions $v\circledast w$ and $v^\prime\circledast w^\prime$ if $2(\dim V+\dim W)-4\leq m$. The bound $m\geq 2(\dim V+\dim W)-4$ is precisely the lower bound established in the first part of this section making this result optimal. 

Furthermore we consider the scenario in which $v\in\mathbb{C}^{n_1}$ is assumed to be $s_1$-sparse in some basis $E\subseteq \mathbb{C}^{n_1}$ and $w\in\mathbb{C}^{n_2}$ is assumed to be $s_2$-sparse in some basis $D\subseteq \mathbb{C}^{n_2}$. Similar to the previous result we show that for generic bases $E,D$ any two pairs of vectors $(v,w),(v^\prime,w^\prime)\in \mathbb{C}^{n_1}\times \mathbb{C}^{n_2}$ consistent with the sparsity constraints can be discriminated from their circular convolution up to the trivial scaling ambiguity if $2(s_1+s_2)-2\leq m$. Again, since the bound $m\geq 2(s_1+s_2)-2$ is precisely the lower bound established in the first part of this section, this result is optimal.

The proofs of our main results are given in Section \ref{proofs}, some measure theoretic technicalities are deferred to Appendix \ref{measure}.

Our techniques also apply to the weak identifiability problem, i.e., the corresponding problem for $(u,v)$ fixed. For this problem, in Appendix \ref{app}, we also slightly improve the conditions derived in \cite{li2015identifiability} in the case of sparsity constraints and show optimality.

\section{Preliminaries}\label{0}
Let us first fix some notation. We denote the Euclidean norm by $\|\cdot\|_2$. By $M(n_1,n_2)$ we denote the set of complex $n_1\times n_2$ matrices. In the following we always assume $n_1,n_2\geq 2$. The transpose (conjugate transpose) of a matrix $X\in M(n_1,n_2)$ is denoted by $X^t$ ($X^*$). We equip $M(n_1,n_2)$ with the Hilbert-Schmidt inner product defined by setting $\langle X,Y\rangle=\text{tr}(X^*Y)$ for all $X,Y\in M(n_1,n_2)$ and by $\|\cdot\|_F$ we denote the Hilbert-Schmidt/Frobenius norm.  For $i\in\{1,\hdots,n_1\},j\in\{1,\hdots,n_2\}$ we denote by $X_i$ the $i$-th row and by $X_{ij}$ the entry in the $i$-th row and $j$-th column of a matrix $X\in M(n_1,n_2)$ . By $\mathds{1}_n$ we denote the identity on $\mathbb{C}^n$. We denote by $M^r(n_1,n_2)\subseteq M(n_1,n_2)$ the set of complex matrices of rank at most $r$. For $0<s\le n$ let $\mathbb{C}^n_s:=\{x\in\mathbb{C}^n:x\text{ is }s\text{-sparse}.\}$. Furthermore let 
\begin{align*}
M^1_{s_1,s_2}(n_1,n_2):=\{uv^t:u\in\mathbb{C}_{s_1}^{n_1},v\in\mathbb{C}_{s_2}^{n_2}\}.
\end{align*}
For two subsets $V,W\subseteq M(n_1,n_2)$ we define the set $V-W:=\{X-Y:X\in V,Y\in W\}$ and we write $\Delta(V)$ as shorthand for $V-V$. By $\mathcal{L}(m)$ we denote the set of linear maps $M:M(n_1,n_2)\to \mathbb{C}^{m}$ and by $\mathcal{B}(m)$ we denote the set of bilinear maps $B:\mathbb{C}^{n_1}\times \mathbb{C}^{n_2}\to\mathbb{C}^m$. Finally, the kernel and the range of a linear map $L$ are denoted by $\text{ker}\,L$ and $\text{ran}\,L$, respectively. 

Using the Hilbert-Schmidt inner product we can identify $\mathcal{L}(m)$ with $(M(n_2,n_1))^m$. Indeed, for every linear map $M\in\mathcal{L}(m)$ there exists $(Y_1,\hdots,Y_m)\in (M(n_2,n_1))^m$ such that for all $X\in M(n_1,n_2)$ we have $M(X)=(\text{tr}(Y_1X),\hdots,\text{tr}(Y_mX))$ and conversely every $Y:=(Y_1,\hdots,Y_m)\in (M(n_2,n_1))^m$ induces a linear map $M_Y\in\mathcal{L}(m)$ by setting
\begin{align}\label{ind}
M_Y(X):=(\text{tr}(Y_1X),\hdots,\text{tr}(Y_mX))
\end{align}
for all $X\in M(n_1,n_2)$.

In bilinear inverse problems the objective is to reconstruct two vectors $u\in\mathbb{C}^{n_1}\setminus\{0\}$, $v\in\mathbb{C}^{n_2}\setminus\{0\}$ form a measurement outcome $z=B(u,v)\in\mathbb{C}^m$ where $B\in\mathcal{B}(m)$ is a bilinear map. In the best case this can be done modulo the trivial ambiguity
\begin{align*}
B(u,v)=B( \lambda u,1/\lambda v),\ \forall\lambda\in\mathbb{C}\setminus \{0\}.
\end{align*}
This ambiguity naturally induces an equivalence relation on $\mathbb{C}^{n_1}\setminus\{0\}\times\mathbb{C}^{n_2}\setminus\{0\}$ and we denote the equivalence class of $(u,v)\in\mathbb{C}^{n_1}\setminus\{0\}\times\mathbb{C}^{n_2}\setminus\{0\}$ by $[(u,v)]$, i.e., $[(u,v)]=[(u^\prime,v^\prime)]$ for some $u^\prime\in\mathbb{C}^{n_1}\setminus\{0\}$, $v^\prime\in\mathbb{C}^{n_2}\setminus\{0\}$ if there exists $\lambda\in\mathbb{C}\setminus\{0\}$ such that $(u,v)=(\lambda u^\prime,1/\lambda v^\prime)$.

This motivates the following notion of identifiability which is basically the same as part 2 of Definition 2.1 in \cite{li2015identifiability}.
\begin{definition}(Strong identifiability modulo scaling.)
A subset $V\subseteq\mathbb{C}^{n_1}\setminus\{0\}\times\mathbb{C}^{n_2}\setminus\{0\}$ is identifiable modulo scaling with respect to a map $B\in\mathcal{B}(m)$,  if $B(u,v)=B(u^\prime,v^\prime)$ for some $(u,v),(u^\prime,v^\prime)\in V$ implies $[(u,v)]=[(u^\prime,v^\prime)]$.
% A bilinear map $B\in\mathcal{B}(m)$ is called strongly $(s_1,s_2)$-identifiable modulo scaling if $B|_{\mathbb{C}^{n_1}_{s_1}\setminus\{0\}\times \mathbb{C}^{n_2}_{s_2}\setminus\{0\}}$ is strongly identifiability modulo scaling.
\end{definition}

This strong notion of identifiability requires that every measurement corresponds to a unique signal. The following weaker notion, basically the same as part 1 of Definition 2.1 in \cite{li2015identifiability}, requires this only for the measurement arising from a given fixed signal.

\begin{definition}(Weak identifiability modulo scaling.)
The restriction of a map $B\in\mathcal{B}(m)$ to a subset $V\subseteq\mathbb{C}^{n_1}\setminus\{0\}\times\mathbb{C}^{n_2}\setminus\{0\}$ is weakly identifiable modulo scaling at $(u_0, v_0)$ if $B(u,v)=B(u_0,v_0)$ for some $(u,v)\in V$ implies $[(u,v)]=[(u_0,v_0)]$. 
%A bilinear map $B\in\mathcal{B}(m)$ is called weakly $(s_1,s_2)$-identifiable modulo scaling at $(u.v)\in \mathbb{C}^{n_1}_{s_1}\setminus\{0\}\times \mathbb{C}^{n_2}_{s_2}\setminus\{0\}$ if $B|_{\mathbb{C}^{n_1}_{s_1}\setminus\{0\}\times \mathbb{C}^{n_2}_{s_2}\setminus\{0\}}$ is weakly identifiable modulo scaling at $(u.v)$.
\end{definition}

\section{Main Results}\label{main}
\subsection{Tight Bounds for the Injectivity Problem}
In the following we first obtain a general lower bound on the number $m\in\mathbb{N}$ for which there exists an $(s_1,s_2)$-injective bilinear map $B\in\mathcal{B}(m)$. Then, in a next step, we show that this lower bound is indeed tight.

It is well-known that there is a one-to-one correspondence between the sets $\mathcal{B}(m)$ and $\mathcal{L}(m)$. Indeed, every bilinear map $B:\mathbb{C}^{n_1}\times \mathbb{C}^{n_2}\to\mathbb{C}^m$ induces a unique linear map $M:M(n_1,n_2)\to\mathbb{C}^m$ such that for all $x\in\mathbb{C}^{n_1},y\in\mathbb{C}^{n_2}$ we have $M(xy^t)=B(x,y)$. Conversely each linear map $M:M(n_1,n_2)\to\mathbb{C}^m$ induces a bilinear map $B:\mathbb{C}^{n_1}\times \mathbb{C}^{n_2}\to\mathbb{C}^m$ by setting $B(x,y):=M(xy^t)$ for all $x\in\mathbb{C}^{n_1},y\in\mathbb{C}^{n_2}$. Throughout the present section we take advantage of this correspondence and work with the set $\mathcal{L}(m)$ of linear maps rather than the set $\mathcal{B}(m)$ of bilinear maps.

\begin{definition}($(s_1,s_2)$-injective.)
A linear map $M\in\mathcal{L}(m)$ is $(s_1,s_2)$-injective if $M|_{M^1_{s_1,s_2}(n_1,n_2)}$ is injective.
\end{definition}
Clearly, a bilinear map $B\in\mathcal{B}(m)$ is $(s_1,s_2)$-injective modulo scaling if and only if the associated linear map $M\in\mathcal{L}(m)$ is $(s_1,s_2)$-injective.

Let us next introduce a notion of stability for $(s_1,s_2)$-injective linear maps.
\begin{definition}(Stability.)
A linear map $M\in\mathcal{L}(m)$ is stably $(s_1,s_2)$-injective if there exists a constant $C>0$ (possibly dependent on all the parameters of $M$) such that $\|M(X)\|_2\geq C\|X\|_F$ for all $X\in\Delta(M^1_{s_1,s_2}(n_1,n_2))$.
\end{definition}
\begin{remark}
Formally, the notion of stability we give here can be understood as a generalization of the notion of stability given in Definition 2.3 of \cite{eldar2014phase} for the Phase Retrieval Problem. Indeed for $x,y\in\mathbb{R}^n$ we have on the one hand 
\begin{align*}
\|x-y\|_2^2\|x+y\|_2^2-\|xx^t-yy^t\|_F^2=2(\|x\|_2^2\|y\|_2^2-\langle x,y\rangle)\geq 0
\end{align*}
by the Cauchy-Schwarz inequality and on the other hand 
\begin{align*}
\|xx^t-yy^t\|=&\sup_{\|v\|_2=1}v^t(xx^t-yy^t)v=\frac12\sup_{\|v\|_2=1}v^t((x+y)(x-y)^t+(x-y)(x+y)^t)v\\
&\leq \sup_{\|v\|_2=1,\|u\|_2=1}u^t((x+y)(x-y)^t)v=\|x-y\|_2^2\|x+y\|_2^2,
\end{align*}
where $\|\cdot\|$ denotes the operator norm.
However, different from \cite{eldar2014phase}, throughout the present paper we do not aim for a universal constant $C$ for which the stability bound holds, but rather allow $C$ to depend on all the parameters of the linear map $M$.
\end{remark}

Under the premise of this rather weak notion of stability we obtain the following lower bound on the number of measurement outcomes necessary for $(s_1,s_2)$-injectivity.
\begin{theorem}\label{thm1}(Lower bound.)
If there is a stably $(s_1,s_2)$-injective linear map $M\in\mathcal{L}(m)$, then 
\begin{align*}
m\ge\left\{
	\begin{array}{ll}
		2(n_1+n_2)-4  & \mbox{if } s_1=n_1,s_2=n_2, \\
		2(s_1+s_2)-2 & \mbox{else}.
	\end{array}
\right.
\end{align*} 
\end{theorem}
The proof of this result can be found in Section \ref{proofs}.
\begin{remark}
Let us note that this lower bound also applies to a slightly more general scenario which will be of relevance in the next part of this section. Indeed, let $V\subseteq\mathbb{C}^{n_1}$, $W\subseteq\mathbb{C}^{n_2}$ be subspaces with bases $E:=\{e_1,\hdots,e_{\dim V}\}\subseteq V$ and $F=\{f_1,\hdots,f_{\dim W}\}\subseteq W$, respectively. Denote by $V_{s_1}\subseteq V$ the elements of $V$ that are $s_1$-sparse in the basis $E$ and by $W_{s_2}\subseteq W$ the elements of $W$ that are $s_2$-sparse in the basis $F$. If a bilinear map $B\in\mathcal{B}(m)$ is such that $B$ restricted to $V_{s_1}\times W_{s_2}$ is injective modulo scaling, then the bilinear map $\tilde{B}:\mathbb{C}^{\dim V}\times\mathbb{C}^{\dim W}\to\mathbb{C}^n,\, (x,y)\mapsto B(\sum_{i}x_ie_i,\sum_{i}y_if_i)$ is $(s_1,s_2)$-injective modulo scaling. Thus, under the premise of stability, the lower bound given in \ref{thm1} also applies to this scenario.
\end{remark}

Next we show that the lower bound given by Theorem \ref{thm1} is indeed tight. In the following theorem the term ``almost all'' refers to the Lebesgue measure on $(M(n_1,n_2))^m$ which represents $\mathcal{L}(m)$ via \eqref{ind}.
\begin{theorem}\label{thma}(Upper bound.)
Almost all linear maps $M\in\mathcal{L}(m)$ are stably $(s_1,s_2)$-injective if 
\begin{align*}
m\ge\left\{
	\begin{array}{ll}
		2(n_1+n_2)-4  & \mbox{if } s_1=n_1,s_2=n_2, \\
		2(s_1+s_2)-2 & \mbox{else}.
	\end{array}
\right.
\end{align*}
\end{theorem}
The following Theorem shows that the lower bound established in Theorem \ref{thm1} also applies for a certain structured subset of $\mathcal{L}(m)$ which will be of relevance in the next section. In the following theorem the term ``almost all'' refers to the Lebesgue measure.
\begin{theorem}\label{thm2}
For almost all $(Y,Z)\in M(m,n_1)\times M(m,n_2)$, the linear map 
\begin{align*}
M_{Y,Z}: M(n_1,n_2)\to\mathbb{C}^m,\ X\mapsto (\text{tr}(Z_1^tY_1X),\hdots,\text{tr}(Z_m^tY_mX))
\end{align*}
is stably $(s_1,s_2)$-injective if 
\begin{align*}
m\ge\left\{
	\begin{array}{ll}
		2(n_1+n_2)-4  & \mbox{if } s_1=n_1,s_2=n_2, \\
		2(s_1+s_2)-2 & \mbox{else}. 
	\end{array}
\right.
\end{align*}
\end{theorem}
The proof of this result can be found in Section \ref{proofs}.

\subsection{Blind Deconvolution with Subspace and Sparsity Constraints}\label{1}

In this section we apply Theorem \ref{thm2} to the deconvolution problem. By $v\circledast w\in\mathbb{C}^m$ we denote the circular convolution of vectors $v,w\in\mathbb{C}^m$, i.e., for all $i\in\{1,\hdots,m\}$ we have $(v\circledast w)_i=\sum_{j=1}^mv_{j}w_{[(i-j-1)\,\text{mod}\,m]+1}$ . The mapping
\begin{align*}
C:\mathbb{C}^m\times\mathbb{C}^m&\to\mathbb{C}^m\\
(v,w)&\mapsto v\circledast w
\end{align*}
is easily seen to be bilinear. Since the dimension of the domain of $C$ is larger than the dimension of the range of $C$, $C$ cannot be injective modulo scaling for $m>1$. However, when imposing sparsity or subspace constraints on the domain of $C$ this can change.

Let us first focus on subspace constraints. For $k\leq m$, consider the set
\begin{align*}
G(m,k):=\{P\in M(m,m): P\text{ is a rank }k\text{ orthogonal projection.}\}
\end{align*}
which can naturally be identified with the Grassmannian of $k$-dimensional subspaces of $\mathbb{C}^m$. We make $(G(m,k),d^\prime)$ a compact metric space by setting $d^\prime(P,P^\prime):=\|P-P^\prime\|_F$ for all $P,P^\prime\in G(m,k)$. Let $\mu_k$ be the Haar measure of $(G(m,k),d^\prime)$ with respect to the action $G:U(m)\times G(m,k)\to G(m,k), (U,P)\mapsto UPU^*$ of the unitary group $U(m)$ on $G(m,k)$ \footnote{Recall that for $(X,d)$ a compact metric space and $G$ be a group of isometries of $(X,d)$ that acts transitively on $X$, the Haar measure $\lambda$ on $(X,d)$ with respect to $G$ is the unique $G$-invariant Borel measure with $\lambda(X)=1$ (See for instance theorems 1.1 and 1.3 of \cite{milman2009asymptotic}).}. In the following theorem, the term ``almost all'' refers to the product measure measure $\mu_k\times\mu_l$.
\begin{theorem}\label{thm3}(Deconvolution with subspace constraint.)
Let $k,l\in\mathbb{N}_+$ be such that $2(k+l)-4\leq m$. Then, for almost all pairs of projections $(P,P^\prime)\in G(m,k)\times G(m,l)$, $\text{ran}\,P\times \text{ran}\,P^\prime$ is strongly identifiable modulo scaling with respect to the circular convolution map $C:\mathbb{C}^m\times\mathbb{C}^m\to\mathbb{C}^m,\ (v,w)\mapsto v\circledast w$.
\end{theorem}
\begin{remarks}
\ \newline
\vspace{-18pt}
\begin{itemize}
\item[(a)] Note, that this result is optimal. Indeed, since $\Delta(M^1(\dim V,\dim W))=M^2(\dim V,\dim W)$, we conclude that $\Delta(M^1(\dim V,\dim W))$ is closed. Thus, by Proposition \ref{prop3}, the map $C$ restricted to $\text{ran}\,P\times \text{ran}\,P$ is injective modulo scaling if and only if the associated linear map $L_C:M(k,l)\to\mathbb{C}^m$ (cf. remark after Theorem \ref{thm1}) is stably $(k,l)$-injective.  Hence, by Theorem \ref{thm1}, there do not exist projections  $(P,P^\prime)\in G(m,k)\times G(m,l)$ such that $C$ restricted to $\text{ran}\,P\times \text{ran}\,P^\prime$ is injective modulo scaling  if $2(k+l)-4> m$. This shows the statement for the second part. For the first part the same argument applies.
\item[(b)] As a comparison, Theorem~3.2 in \cite{li2015identifiability} requires $2(k+l)< m$ measurements to guarantee strong identifiability, so as expected there is only a small improvement. We hence consider it to be our main achievement that our bounds are provably optimal.
\end{itemize}
\end{remarks}
The proof of this theorem is given in Section \ref{proofs}.

Next we consider sparsity constraints. Denote by 
\begin{align*}
F(m,k):=\{X\in M(m,k):\text{rank}\,X=k\}
\end{align*}
the set of all collections of $k$ linearly independent vectors in $\mathbb{C}^m$.  In the following theorem the term ``almost all'' refers to the Lebesgue measure.

\begin{theorem}\label{thm4}(Deconvolution with sparsity constraint.)
For  $E\in F(m,k)$, set $\text{ran}(E)_s:=\{x\in\text{ran}\,E:x\text{ is }s\text{-sparse when expanded in }E .\}$. Let $s_1,s_2\in\mathbb{N}_+$ be such that $2 (s_1+s_2)-2\leq m$ and let $k,l\in\mathbb{N}$ be such that $s_1<k\leq m$, $s_2<l\leq m$. Then, for almost all pairs  $(E,D)\in F(m,k)\times F(m,l) $, $\text{ran}(E)_{s_1}\times \text{ran}(D)_{s_2}$ is strongly identifiable modulo scaling with respect to the circular convolution map $C:\mathbb{C}^m\times\mathbb{C}^m\to\mathbb{C}^m,\ (v,w)\mapsto v\circledast w$.
\end{theorem}
\begin{remarks}
\ \newline
\vspace{-18pt}
\begin{itemize}
\item[(a)]
Note that under the premise of stability this result is optimal. Indeed,  by Theorem \ref{thm1} and the remark afterwards there do not exist $(E,D)\in F(m,k)\times F(m,l)$ such that $C$ restricted to $\text{ran}(E)_{s_1}\times \text{ran}(F)_{s_2}$ is identifiable modulo scaling  if $2(s_1+s_2)-2> m$.
\item[(b)] As a comparison, Theorem~3.2 in \cite{li2015identifiability} requires $m>2(s_1+s_2)$ measurements for strong identifiability in this case.
\item[(c)]
Note that this result also applies to the standard convolution $\tilde{C}:\mathbb{C}^n\times\mathbb{C}^n\to\mathbb{C}^{2n-1},\,(u,v)\mapsto (\sum_{j=-\infty}^{+\infty}u_jv_{i-(j-1)})_{i=1}^{2n-1}$, where $u_j=0$, $v_j=0$ for $j\not\in\{1,\hdots,n\}$. Indeed, by embedding $\mathbb{C}^n$ appropriately in $\mathbb{C}^{2n-1}$, $\tilde{C}$ can be understood as a restriction of the circular convolution $C:\mathbb{C}^{2n-1}\times \mathbb{C}^{2n-1}\to \mathbb{C}^{2n-1}$.
\end{itemize}
\end{remarks}
The proof of this theorem is given in Section \ref{proofs}.

\section{Proofs}\label{proofs}
\subsection{Algebraic Geometry Background and Notation}
The proof of Theorem \ref{thm1} relies on results from classical algebraic geometry so let us fix some conventions (which are close to \cite{hartshorne1977algebraic}). We call a set $V\subseteq\mathbb{C}^n$ an algebraic set if it is the common zero locus of a set of complex polynomials in $n$ variables. The Zariski topology on $\mathbb{C}^n$ is defined by choosing its closed sets to be the algebraic sets. A non-empty subset $V$ of $\mathbb{C}^n$ equipped with the Zariski topology is called irreducible if it cannot be expressed as a union of two proper subsets of $V$, each of which is relatively closed in $V$.  We call an algebraic set an affine variety if it is irreducible. Subsets of an algebraic set that are relatively open in the Zariski topology are called quasi algebraic sets. For a subset $V\subseteq\mathbb{C}^n$ , we denote by $P(V)\subseteq P(\mathbb{C}^{n})$ its projectification, i.e., the image of $V$ under the canonical projection $P:\mathbb{C}^n\to P(\mathbb{C}^{n})$. If a subset $V\subseteq\mathbb{C}^n$ is the common zero locus of a set of homogeneous polynomials,  we call $P(V)$ a projective algebraic set. We denote by $\overline{V}$ the analytic closure of a subset $V\subseteq \mathbb{C}^n$, i.e., its closure in the standard topology of $\mathbb{C}^n$. Furthermore $\overline{V}_Z$ denotes the closure of a subset $V\subseteq \mathbb{C}^n$ in the Zariski topology. By $\dim V$ we denote the algebraic dimension of a subset $V\subseteq\mathbb{C}^n$.

For a subset $A\subseteq\{1,\hdots,n\}$ define the projection $P_A:\mathbb{C}^n\to\mathbb{C}^n$ by setting
\begin{align*}
(P_A(x))_i:= \left\{
	\begin{array}{ll}
		x_i  & \mbox{if } i\in A, \\
		0 & \mbox{else} 
	\end{array}
\right.
\end{align*}
for all $x\in\mathbb{C}^n$ and $i\in\{1,\hdots,n\}$.
 Furthermore, let $\mathcal{A}(n,s):=\{A\subseteq\{1,\hdots,n\}:|A|=s\}$. Then 
\begin{align*}
M^1_{s_1,s_2}(n_1,n_2)=\bigcup_{A\in \mathcal{A}(n_1,s_1),\ B\in \mathcal{A}(n_2,s_2)}W_{A,B},
\end{align*}
where $W_{A,B}:=\{P_Au(P_Bv)^*:u\in \mathbb{C}^{n_1},\ v\in \mathbb{C}^{n_2}\}$. Let $S_{s_1,s_2}:=W_{\{1,\hdots,s_1\},\{1,\hdots,s_2\}}$.

\subsection{Proof of Theorem \ref{thm1}}\label{sub1}
Theorem \ref{thm1} can be proven straightforwardly from the following three propositions. Their proofs are relegated to the end of this subsection.
\begin{proposition}\label{prop1}
We have 
\begin{align*}
\dim \Delta(M^1_{s_1,s_2}(n_1,n_2))=\left\{
	\begin{array}{ll}
		2(n_1+n_2-2) & \mbox{if } s_1=n_1,s_2=n_2, \\
		2(s_1+s_2-1) & \mbox{else}.
	\end{array}
\right.
\end{align*}
\end{proposition}

\begin{proposition}\label{prop2}
The analytic closure of $\Delta(M^1_{s_1,s_2}(n_1,n_2))$ is the common zero locus of a set of homogeneous polynomials. In particular $\overline{\Delta(M^1_{s_1,s_2}(n_1,n_2))}=\overline{\Delta(M^1_{s_1,s_2}(n_1,n_2))}_Z$.
\end{proposition}

\begin{proposition}\label{prop3}
A linear map $M\in\mathcal{L}(m)$ is stably $(s_1,s_2)$-injective if and only if $P(\text{ker}\,M)\cap P(\overline{\Delta(M^1_{s_1,s_2}(n_1,n_2))})=\emptyset$. 
\end{proposition}

\begin{proof}[Proof of Theorem \ref{thm1}.]
By Proposition \ref{prop3} a linear map $M\in\mathcal{L}(m)$ is stably $(s_1,s_2)$-injective if and only if $P(\text{ker}\,M)\cap P(\overline{\Delta(M^1_{s_1,s_2}(n_1,n_2))})=\emptyset$. Clearly $P(\text{ker}\,M)$ is a projective algebraic set and by Proposition \ref{prop2}, $P(\overline{\Delta(M^1_{s_1,s_2}(n_1,n_2))})$ also is a projective algebraic set.

By the intersection theorem for complex varieties (see Theorem 7.2 of \cite{hartshorne1977algebraic}), together with the observation that a projective algebraic set contains an irreducible projective algebraic subset of the same dimension, two projective algebraic sets $V,W\subseteq P(\mathbb{C}^{n+1})$ have non-empty intersection if $\dim V+\dim W\geq n$. Hence, if $M$ is stably $(s_1,s_2)$-injective we have $\dim P(\text{ker}\,M)+\dim P(\overline{\Delta(M^1_{s_1,s_2}(n_1,n_2))})<\dim P(\mathbb{C}^{n_1n_2})$. Noting that $\dim P(\text{ker}\,M)\ge n_1n_2-m-1$ \footnote{Note that if $V\subseteq\mathbb{C}^n$ is such that $P(V)$ is a projective algebraic set, then $\dim P(V)=\dim V-1$.}, we find $n_1n_2-m-1+\dim P(\overline{\Delta(M^1_{s_1,s_2}(n_1,n_2))})<n_1n_2-1$. But this implies, using again Proposition \ref{prop2}, that 
\begin{align*}
m&\ge\dim P(\overline{\Delta(M^1_{s_1,s_2}(n_1,n_2))})+1=\dim \overline{\Delta(M^1_{s_1,s_2}(n_1,n_2))}\\
&=\dim \overline{\Delta(M^1_{s_1,s_2}(n_1,n_2))}_Z=\dim \Delta(M^1_{s_1,s_2}(n_1,n_2))
\end{align*}
and Proposition \ref{prop1} concludes the proof.
\end{proof}

\begin{proof}[Proof of Proposition \ref{prop1}.]
Let $A\in \mathcal{A}(n_1,s_1),\ B\in \mathcal{A}(n_2,s_2)$. We begin by determining the dimension of $W_{A,B}$. $W_{A,B}$ is the set of $X\in M^1(n_1,n_2)$ such that $((\id_{n_1}-P_A)X)_{ij}=0$, $(X(\id_{n_2}-P_B))_{ij}=0$ for all $i\in\{1,\hdots,n_1\}$, $j\in\{1,\hdots,n_2\}$. Furthermore, $M^1(n_1,n_2)$ is well-known to be an algebraic set \footnote{$M^1(n_1,n_2)$ is the common zero locus of the $2\times2$-minors.} and hence $W_{A,B}$ is an algebraic set. Let $I_A:\mathbb{C}^{s_1}\to\text{ran}\,P_A$, $I_B:\mathbb{C}^{s_2}\to\text{ran}\,P_B$ be the canonical linear embeddings, then the mapping 
\begin{align}\label{eta}
\eta:M(s_1,s_2)\to W_{A,B},\,X\to I_AXI_B
\end{align}
 also is a linear embedding. In particular $\eta|_{M^1(s_1,s_2)}$ yields an isomorphism between $M^1(s_1,s_2)$ and $W_{A,B}$. Hence, using Example 12.1 of \cite{harris2013algebraic}, we find $\dim W_{A,B}=\dim M^1(s_1,s_2)=s_1+s_2-1$.

We now start by considering the case $n_1=s_1$ and $n_2=s_2$. In this case we find $\Delta(M^1_{n_1,n_2}(n_1,n_2))=\Delta(M^1(n_1,n_2))=M^2(n_1,n_2)$. The set $M^2(n_1,n_2)$ is an algebraic set and its dimension is given by $2(n_1+n_2-2)$ (Again by Example 12.1 of \cite{harris2013algebraic}.).

Secondly, let us assume that $s_1<n_1$, the case $s_2<n_2$ can be treated analogously. Let $A,A^\prime\in \mathcal{A}(n_1,s_1),\ B,B^\prime\in \mathcal{A}(n_2,s_2)$. Consider the morphism
\begin{align*}
\psi: W_{A,B}\times W_{A^\prime,B^\prime}&\to M(n_1,n_2)\\
(X,Y)&\to X-Y
\end{align*}
and note that $\psi(W_{A,B}\times W_{A^\prime,B^\prime})=W_{A,B}-W_{A^\prime,B^\prime} $. Therefore we find 
\begin{align*}
\dim (W_{A,B}-W_{A^\prime,B^\prime})\leq\dim W_{A,B}+\dim W_{A^\prime,B^\prime}=2\dim S=2(s_1+s_2-1)
\end{align*}
and since $\Delta(M^1_{s_1,s_2}(n_1,n_2))=\bigcup_{A,A^\prime\in \mathcal{A}(n_1,s_1),\ B,B^\prime\in \mathcal{A}(n_2,s_2)}W_{A,B}-W_{A^\prime,B^\prime}$ this implies 
\begin{align}\label{ieq1}
\dim \Delta(M^1_{s_1,s_2}(n_1,n_2))\leq 2(s_1+s_2-1).
\end{align}

Since $s_1<n_1$ there exist $A,A^\prime\in \mathcal{A}(n_1,s_1)$ such that $1\in A$, $1\notin A^\prime$ and $2\in A^\prime$, $2\notin A$. Furthermore, let $B\in \mathcal{A}(n_2,s_2)$ and consider the set
\begin{align*}
\mathcal{D}:=&W_{A,B}\times W_{A^\prime,B}\cap \{(X,Y)\in M(n_1,n_2)\times M(n_1,n_2): X_{11}\neq 0,\ Y_{21}\neq 0,\\
&Y_{21}\cdot X_1-X_{11}\cdot Y_2\neq 0\}.
\end{align*}
The set $\mathcal{D}$ clearly is non-empty and quasi algebraic. Furthermore we have $\dim \mathcal{D}=2(s_1+s_2-1)$. One way to see this is the following: Since both $W_{A,B}$ and $W_{A^\prime,B}$ are isomorphic to $M^1(s_1,s_2)$, $W_{A,B}\times W_{A^\prime,B}$ is isomorphic to $M^1(s_1,s_2)\times M^1(s_1,s_2)$.  By Proposition 12.2 of \cite{harris2013algebraic} $M^1(s_1,s_2)$ is irreducible \footnote{Proposition 12.2 of \cite{harris2013algebraic} just states that $P(M^1(s_1,s_2))$ is irreducible. An algebraic set is irreducible if and only if its associated polynomial ideal is prime and the same holds for projective algebraic sets (see for instance Corollary 1.4 and Exercise 2.4 (b) of \cite{hartshorne1977algebraic}). But the polynomial ideals associated to  $PM^1(s_1,s_2)$ and  $M^1(s_1,s_2)$ clearly are equal and thus $M^1(s_1,s_2)$ is irreducible as well.} and hence an affine variety. By Exercice 3.15 of \cite{hartshorne1977algebraic}, the product of two affine varieties is irreducible and hence $M^1(s_1,s_2)\times M^1(s_1,s_2)$ is irreducible. Finally, by Example 1.1.3 of \cite{hartshorne1977algebraic} a non-empty quasi algebraic subset of an irreducible set is irreducible. Hence $\mathcal{D}$ is irreducible and thus by Exercise 1.6 and Proposition 1.10 of \cite{hartshorne1977algebraic} we have $\dim \mathcal{D}=\dim \overline{\mathcal{D}}_Z=\dim W_{A,B}\times W_{A^\prime,B}=2(s_1+s_2-1)$
 \footnote{A more direct approach consists of showing the injectivity of the mapping
 \begin{align*}
&(\mathbb{C}^{s_1}\setminus\{0\}\times \mathbb{C}^{s_2-1}\setminus\{0\})\times(\mathbb{C}^{s_1}\setminus\{0\}\times \mathbb{C}^{s_2-1}\setminus\{0\})\setminus\mathcal{S}\to\mathcal{D},\\
&((u,v),(u^\prime,v^\prime))\mapsto \left(I_A(u)I_B\left(\begin{pmatrix} 1\\ v\end{pmatrix}\right)^t,I_{A^\prime}(u^\prime)I_B\left(\begin{pmatrix} 1\\ v^\prime\end{pmatrix}\right)^t\right)
\end{align*}
where $\mathcal{S}:=\{((u,v),(u^\prime,v^\prime))\in(\mathbb{C}^{s_1}\setminus\{0\}\times \mathbb{C}^{s_2-1}\setminus\{0\})\times (\mathbb{C}^{s_1}\setminus\{0\}\times \mathbb{C}^{s_2-1}\setminus\{0\}):v_1-v^\prime_1=0,\, u_1=0,\,u_2^\prime=0\}$ and $I_A$ denotes a linear embedding of $\mathbb{C}^s_1$ into $\text{ran}\, P_A$. We leave details to the reader.}. 

Next we prove that $\dim\psi(\mathcal{D})=\dim \mathcal{D}=2(s_1+s_2-1)$ by showing that the morphism $\psi|_{\mathcal{D}}$ is injective: Let  $(X,Y)\in\mathcal{D}$. Then, by the definition of $\mathcal{D}$, $X_1$ and $Y_2$ are non-vanishing and linearly independent. Hence there are vectors $\omega_1(X),\omega_2(Y)\in\mathbb{C}^{n_2}$ such that $\langle\omega_1(X),X_1\rangle=1,\ \langle\omega_2(Y),X_1\rangle=0, \langle\omega_1(X),Y_2\rangle=0$ and $\langle\omega_2(Y),Y_2\rangle=1$. Note that $\langle\omega_1(X),Y_i\rangle=0$ for all $i\in\{1,\hdots,n_1\}$ by the fact that $Y$ is rank one. Similarly we have $\langle\omega_2(Y),X_i\rangle=0$ for all $i\in\{1,\hdots,n_1\}$. Again using the fact that both $X$ and $Y$ are rank one it follows from a straightforward computation that for all $i\in\{1,\hdots,n_1\}$ we have
\begin{align*}
\langle\omega_1(X),(X-Y)_i\rangle X_1&=\langle\omega_1(X),X_i\rangle X_1=X_i\\
-\langle\omega_2(Y),(X-Y)_i\rangle Y_2&=\langle\omega_2(Y),Y_i\rangle Y_2=Y_i.
\end{align*}
This explicitly defines an inverse map of $\psi$ on $\psi(\mathcal{D})$, showing that $\psi|_D$ is injective.

But since $\psi(D)\subseteq\Delta(M^1_{s_1,s_2}(n_1,n_2))=\bigcup_{A,A^\prime\in \mathcal{A}(n_1,s_1),\ B,B^\prime\in \mathcal{A}(n_2,s_2)}W_{A,B}-W_{A^\prime,B^\prime}$ we find $2(s_1+s_2-1)=\dim\mathcal{D}=\dim\psi\mathcal(D)\leq \dim \Delta(M^1_{s_1,s_2}(n_1,n_2))$, which, together with Inequality \eqref{ieq1}, concludes the proof.
%Clearly $\pi|_S$ is injective and thus one can identify $S$ with $\pi(S)=M_1(s_1,s_2)$. Furthermore, by %linearity of $\pi$, we can identify $\Delta(S)$ with $\pi(\Delta(S))=\Delta(M_1(s_1,s_2))=M_2(s_1,s_2)$. The %determinental variety $M_2(s_1,s_2)$ is the common zero locus of the set of $3\times 3$ minors of an $s_1\times %s_2$ matrix. These are homogenous polynomials and furthermore we have $\dim P(M_2(s_1,s_2))=2(s_1+s_2-2)-1$ %\footnote{See Example 12.1 in \cite{harris2013algebraic}.}. This proves the claim for $S$.
%
%But for any $A\in \mathcal{A}(n_1,s_1),\ B\in \mathcal{A}(n_2,s_2)$ there clearly are permutation matrices $P_1\in M(n_1,n_1),\ %P_2\in M(n_2,n_2)$ such that the mapping
%\begin{align*}
%\phi:S&\to W_{A,B}\\
%X&\mapsto P_1XP_2
%\end{align*}
%is bijective. Hence $W_{A,B}\cong S$ and this concludes the proof.
\end{proof}

\begin{proof}[Proof of Proposition \ref{prop2}.]
Since \begin{align*}
\overline{\Delta(M^1_{s_1,s_2}(n_1,n_2))}=\bigcup_{A,A^\prime\in \mathcal{A}(n_1,s_1),\ B,B^\prime\in \mathcal{A}(n_2,s_2)}\overline{W_{A,B}-W_{A^\prime,B^\prime}},
\end{align*}
it suffices to prove that $\overline{W_{A,B}-W_{A^\prime,B^\prime}}$ is the common zero locus of a set of homogeneous polynomials for all $A,A^\prime\in \mathcal{A}(n_1,s_1),\ B,B^\prime\in \mathcal{A}(n_2,s_2)$.

So let $A,A^\prime\in \mathcal{A}(n_1,s_1),\ B,B^\prime\in \mathcal{A}(n_2,s_2)$. By the first paragraph in the proof of Theorem 3.16 of \cite{harris2013algebraic}, $\psi(W_{A,B}\times W_{A^\prime,B^\prime})=W_{A,B}-W_{A^\prime,B^\prime}\subseteq M(n_1,n_2)$ contains a non-empty quasi algebraic subset of its Zarisiki closure \footnote{Here, we use the fact that $\mathbb{C}^n$ is homeomorphic to $P(\mathbb{C}^{n+1})\setminus\{x_0=0\}$ (See Proposition 2.2 in \cite{hartshorne1977algebraic}.).}. Consequently,  the analytic closure of $W_{A,B}-W_{A^\prime,B^\prime}$ coincides with its Zariski closure by Theorem 1 in Chapter 1.10 of \cite{mumford1999red}. Hence the analytic closure of $W_{A,B}-W_{A^\prime,B^\prime}$ is the common zero locus of a finite set of polynomials $\{p_i\}_{i\in I}$.

Let $X\in \overline{W_{A,B}-W_{A^\prime,B^\prime}}$. We now show that $\lambda X\in\overline{W_{A,B}-W_{A^\prime,B^\prime}}$ for all $\lambda\in\mathbb{C}$: There exists a sequence $(X_n)_{n\in\mathbb{N}}\subseteq W_{A,B}-W_{A^\prime,B^\prime}$ that converges to $X$. Next observe that if $Y\in W_{A,B}-W_{A^\prime,B^\prime}$ we also have $\lambda Y\in W_{A,B}-W_{A^\prime,B^\prime}$ for all $\lambda\in\mathbb{C}$. Now let $\lambda\in\mathbb{C}$ and observe that the sequence $(\lambda X_n)_{n\in\mathbb{N}}\subseteq W_{A,B}-W_{A^\prime,B^\prime}$ converges to $\lambda X$ and thus $\lambda X\in\overline{W_{A,B}-W_{A^\prime,B^\prime}}$.

Finally we just have to show that $\overline{W_{A,B}-W_{A^\prime,B^\prime}}$ is the common zero locus of a set of homogeneous polynomials. Let $i\in I$ and let $d_i$ be the degree of $p_i$. Consider the decomposition $p_i=\sum_{j=0}^{d_i}p_{i,j}$  where the $p_{i,j}$ are homogeneous polynomials of degree $j$. Let $X\in\overline{W_{A,B}-W_{A^\prime,B^\prime}}$. Then we have for all $\lambda\in\mathbb{C}$
\begin{align*}
0=p_i(\lambda X)=\sum_{j=0}^d\lambda^jp_{i,j}(X).
\end{align*}
Since this holds for all $\lambda\in\mathbb{C}$, we conclude that $p_i(X)=0$ if and only if $p_{i,j}(X)=0$ for all $j\in\{1,\hdots,d\}$. Repeating this for all $i\in I$ we find a set $J:=\{p_{i,j}\}_{i\in I,j\in\{1,\hdots,d_i\}}$ of homogenous polynomials such that $\overline{W_{A,B}-W_{A^\prime,B^\prime}}$ is the common zero locus of $J$.
\end{proof}

\begin{proof}[Proof of Proposition \ref{prop3}.]
Assume for a contradiction that there is $X\in M(n_1,n_2)$ with $\|X\|_F=1$ such that $P(X)\in P(\text{ker}\,M)\cap P(\overline{\Delta(M^1_{s_1,s_2}(n_1,n_2))})$. Then there is a sequence $(X_n)_{n\in\mathbb{N}}\subseteq \Delta(M^1_{s_1,s_2}(n_1,n_2))$ with $\|X_n\|_F=1$ for all $n\in\mathbb{N}$ that converges to $X$. Thus, by the continuity of $M$, the sequence $\|M(X_n)\|_2$ converges to $\|M(X)\|_2=0$. In particular for any $c>0$ there is an $N\in\mathbb{N}$ such that $\|M(X_N)\|_2\leq c$.

Conversely set $C=\min_{Y\in O}\|M(Y)\|_2$ where $O:=\{Y\in\overline{\Delta(M^1_{s_1,s_2}(n_1,n_2))}:\|Y\|_F=1\}$. If $P(\text{ker}\,M)\cap P(\overline{\Delta(M^1_{s_1,s_2}(n_1,n_2))})=\emptyset$ we have $\|M(Y)\|_2>0$ for all $Y\in O$ and by compactness of $O$ we conclude $C>0$.
\end{proof}

\subsection{Proofs of Theorems \ref{thma} and \ref{thm2} }
The approach we take in this section is similar to the approach taken in \cite{balan2006signal} to prove injectivity for the phase retrieval problem. The proofs of the theorems are immediate consequence of the following two proposition. Their proof is very close to the proof of Proposition III.1 in \cite{kech2}.
\begin{proposition}\label{propk}
Let 
\begin{align}
m\ge\left\{
	\begin{array}{ll}
		2(n_1+n_2)-4  & \mbox{if } s_1=n_1,s_2=n_2, \\
		2(s_1+s_2)-2 & \mbox{else}.
	\end{array}
\right.
\end{align}
Then, the set of $(Y,Z)\in M(m,n_1)\times M(m,n_2)$ such that the linear map 
\begin{align*}
M_{Y,Z}: M(n_1,n_2)\to\mathbb{C}^m,\ X\mapsto (\text{tr}(Z_1^tY_1X),\hdots,\text{tr}(Z_m^tY_mX))
\end{align*}
is not stably $(s_1,s_2)$-injective has strictly smaller dimension than the set $M(m,n_1)\times M(m,n_2)$.
\end{proposition}
\begin{remark}
This proposition directly implies Theorem \ref{thm2}.
\end{remark}
\begin{proof}
Consider the quasi algebraic set
\begin{align*}
\mathcal{W}:=\bigcup_{i\in\{1,\hdots,n_1\},j\in\{1,\hdots,n_2\}}\overline{\Delta(M^1_{s_1,s_2}(n_1,n_2))}\cap \{X\in M(n_1,n_2):X_{ij}=1\}.
\end{align*} Intuitively, the set $\mathcal{W}$ is the union of the canonical charts of $P(\overline{\Delta(M^1_{s_1,s_2}(n_1,n_2))})$ and hence we find $P(\mathcal{W})=P(\overline{\Delta(M^1_{s_1,s_2}(n_1,n_2))})$ and furthermore, using Proposition \ref{prop1},
\begin{align*}
\dim \mathcal{W}=\dim P(\overline{\Delta(M^1_{s_1,s_2}(n_1,n_2))})=\dim \Delta(M^1_{s_1,s_2}(n_1,n_2))-1<m.
\end{align*} For $(Y,Z)\in M(m,n_1)\times M(m,n_2)$ and $X\in M(n_1,n_2)$ define the polynomials
\begin{align}\label{eq1}
p_i(Y,Z,X):=\text{tr}(Z_i^tY_iX),\ i\in\{1,\hdots,m\}.
\end{align}
By $\mathcal{V}$ we denote the common zero locus of the polynomials $\{p_i\}_{i\in\{1,\hdots,m\}}$. Now consider the algebraic set 
\begin{align*}
\mathcal{D}:=(M(m,n_1)\times M(m,n_2)\times\mathcal{W})\cap\mathcal{V}
\end{align*}
and let $\pi:M(m,n_1)\times M(m,n_2)\times\mathcal{W}\to M(m,n_1)\times M(m,n_2)$ be the projection on the factor $M(m,n_1)\times M(m,n_2)$. Let 
\begin{align*}
\mathcal{N}:=\{(Y,Z)\in M(m,n_1)\times M(m,n_2): M_{Y,Z} \text{ is not stably } (s_1,s_2)\text{-injective.}\},
\end{align*}
then we have $\mathcal{N}\subseteq\pi(\mathcal{D})$ \footnote{We even have $\mathcal{N}=\pi(\mathcal{D})$. This, however, will not be of relevance for our argument.}. Indeed, let $(Y,Z)\in\mathcal{N}$. Then, by Proposition \ref{prop3}, there exists $Q\in P(\text{ker}\,M_{Y,Z})\cap P(\overline{\Delta(M^1_{s_1,s_2}(n_1,n_2))})$. Since $P(\mathcal{W})=P(\overline{\Delta(M^1_{s_1,s_2}(n_1,n_2))})$, there exists an $X\in\mathcal{W}$ such that $P(X)=Q$. But then, by linearity of $M_{Y,Z}$, we have $M_{Y,Z}(X)=0$, i.e., $(Y,Z,X)\in\mathcal{D}$. Consequently we have $(Y,Z)\in\pi(\mathcal{D})$.

We will assume for now and show later that $\dim\mathcal{D}=\dim M(m,n_1)+\dim  M(m,n_2)+\dim\mathcal{W}-m$. Then, using $m> \dim\mathcal{W}$, we find that
\begin{align*}
\dim\pi(\mathcal{D})&\leq \dim\mathcal{D}=\dim M(m,n_1)+\dim M(m,n_2)+\dim\mathcal{W}-m\\
&<\dim M(m,n_1)+\dim M(m,n_2).
\end{align*}
That is, $\pi(\mathcal{D})\subseteq M(m,n_1)\times M(m,n_2)$ has strictly smaller dimension than $M(m,n_1)\times M(m,n_2)$ (and thus has Lebesgue measure zero in $M(m,n_1)\times M(m,n_2)$).

Hence, to conclude the proof, it suffices to show that indeed $\dim\mathcal{D}=\dim M(m,n_1)+ \dim M(m,n_2)+\dim\mathcal{W}-m$. To show this, it suffices to prove that for fixed $X\in \mathcal{W}$ the equations $\{p_i=0\}_{i\in\{1,\hdots,m\}}$ reduce the dimension of $M(m,n_1)\times M(m,n_2)$ by $m$ (cf. \cite{balan2006signal}). But for fixed $X\in \mathcal{W}$, the $i$-th equation of \eqref{eq1} just involves the variables of the $i$-th factor of $(\mathbb{C}^{n_1}\times \mathbb{C}^{n_2})^m\simeq M(m,n_1)\times M(m,n_2)$. Hence it suffices to prove that for fixed $X\in \mathcal{W}$ the equation
\begin{align}\label{eq2}
\text{tr}(vw^tX)=0,\ v\in \mathbb{C}^{n_2},w\in \mathbb{C}^{n_1},
\end{align}
reduces the dimension of $\mathbb{C}^{n_1}\times \mathbb{C}^{n_2}$ by one. But Equation \eqref{eq2} is a non-trivial algebraic equation on $\mathbb{C}^{n_1}\times \mathbb{C}^{n_2}$ because $X\neq 0$ and $M(n_1,n_2)$ has a basis of rank one operators. Hence, by Proposition 1.13 of \cite{hartshorne1977algebraic}, Equation \eqref{eq2} reduces the dimension of $\mathbb{C}^{n_1}\times \mathbb{C}^{n_2}$ by one.
\end{proof}

\begin{proposition}
Let 
\begin{align*}
m\ge\left\{
	\begin{array}{ll}
		2(n_1+n_2)-4  & \mbox{if } s_1=n_1,s_2=n_2, \\
		2(s_1+s_2)-2 & \mbox{else}.
	\end{array}
\right.
\end{align*}
Then, the set of $Y:=(Y_1,\hdots,Y_m)\in (M(n_2,n_1))^m$ such that the linear map 
\begin{align*}
M_{Y}: M(n_1,n_2)\to\mathbb{C}^m,\ X\mapsto (\text{tr}(Y_1X),\hdots,\text{tr}(Y_mX))
\end{align*}
is not stably $(s_1,s_2)$-injective has smaller dimension than the set $(M(n_n,n_1))^m$.
\end{proposition}
\begin{remark}
This proposition directly implies Theorem \ref{thma}.
\end{remark}
\begin{proof}
Instead of the polynomials defined in Equation \eqref{eq1}, now consider the polynomials 
\begin{align*}
q_i(Y,X):=\text{tr}(Y_iX),\ i\in\{1,\hdots,m\}
\end{align*}
in $(Y_1,\hdots,Y_m)\in (M(n_2,n_1))^m$ and $X\in M(n_1,n_2)$.

All the arguments given in the remainder of the proof of Proposition \ref{propk} are also valid for these polynomials. Thus the proof can be concluded by going along the lines of the proof of Proposition \ref{propk}.
\end{proof}

\subsection{Proof of theorems \ref{thm3} and \ref{thm4}}
Denote by $F:=\left(\frac{1}{\sqrt{m}}e^{2i\pi\frac{kl}{m}}\right)_{k,l=1}^m\in M(m,m)$ the discrete Fourier matrix. Then we have the following well-known identity
\begin{align}\label{eqF}
v\circledast w=mF^*\left((Fv)\odot (Fw)\right), \forall v,w\in\mathbb{C}^m,
\end{align}
where $\odot$ denotes the Hadamard product, i.e., for $a,b\in\mathbb{C}^m$ we have $a\odot b=(a_ib_i)_{i=1}^m$. 

\begin{proof}[Proof of Theorem \ref{thm4}.]
Since $F$ is invertible it suffices to show that the map $C^\prime:\mathbb{C}^k\times \mathbb{C}^l\to \mathbb{C}^m,\ (u,v)\mapsto (FEu)\odot (FDv)$ is injective modulo scaling for Lebesgue almost all $(E,D)\in M(m,k)\times M(m,l)$. Let $u\in\mathbb{C}^k$, $v\in\mathbb{C}^l$ and $(E,D)\in M(m,k)\times M(m,l)$, then 
\begin{align*}
(C^\prime(u,v))_i=(FEu)_i ( FDv)_i=e_i^tFEuv^tD^tF^te_i=\text{tr}\left([(FD)_i]^t(FE)_iuv^t\right),
\end{align*}
where $\{e_i\}_{i\in\{1,\hdots,m\}}$ denotes the standard orthonormal basis of $\mathbb{C}^m$. Theorem \ref{thm3} implies that $C^\prime$ is injective modulo scaling for Lebesgue almost all $(FE,FD)\in M(m,k)\times M(m,l)$ and hence, since the Lebesgue measure $\lambda$ is unitarily invariant, also for Lebesgue almost all $(E,D)\in M(m,k)\times M(m,l)$.
\end{proof}

Finally, the proof of Theorem \ref{thm3} proceeds identically as the proof of Theorem \ref{thm4}. It only remains to check that in the process of converting the statement of Theorem \ref{thm4} into a statement about the Haar measure on Grassmannians, the set $\mathcal{I}:=\{(E,D)\in M(m,k)\times M(m,l): C^\prime \text{ is injective modulo scaling.}\}$ is mapped to a full measure set. However, since this part of the argument is mainly technical, we relegate it to Appendix \ref{measure}.

\section*{Acknowledgements}
We thank Kiryung Leee and Dustin Mixon for helpful comments. This work was inspired by the workshop Frames and Algebraic \& Combinatorial Geometry at the University of Bremen. F.K.'s contribution was supported by the German Science Foundation (DFG) in the context of the Emmy Noether Junior Research Group `Randomized Sensing and Quantization of Signals and Images'' (KR 4512/1-1) and the project ``Bilinear Compressed Sensing'' (KR 4512/2-1).

\bibliographystyle{unsrt}
\bibliography{bibliography}

\appendix
\section{Completition of the Proof of Theorem \ref{thm3}}\label{measure}
Let $S(m,k):=\{Y\in M(m,k):\|Y\|_F=1\}\subseteq M(m,k)$ be the unit sphere in $M(m,k)$. We make $(S(m,k),d)$ a compact metric space by setting $d(X,Y):=\|X-Y\|_F$ for all $X,Y\in S(m,k)$. Let $U(m,k)$ be the group of isometries of $(S(m,k),d)$ and let $\sigma_k$ be the Haar measure on $S(m,k)$ with respect to $U(m,k)$ . Let
\begin{align*}
\pi_k:M(m,k)\setminus \{0\}\to S(m,k),\ X\mapsto \frac{X}{\|X\|_F}.
\end{align*}
Furthermore let $\pi:M(m,k)\setminus \{0\}\times M(m,l)\setminus \{0\}\to S(m,k)\times S(m,l),\ (X,Y)\mapsto(\pi_k(X),\pi_l(Y))$. It is well-known that for all Borel sets $A\in\mathcal{B}(S(m,k))$ we have $\sigma_k(A)=\lambda(\pi_k^{-1}(A)\cap B_{m,k})/\lambda(B_{m,k})$ where $B_{m,k}=\{X\in M(m,k):\|X\|_F\leq 1\}$ is the unit ball in $M(m,k)$. But the mapping $\pi:M(m,k)\setminus \{0\}\times M(m,l)\setminus \{0\}\to S(m,k)\times S(m,l),\ (X,Y)\mapsto(X/\|X\|_F,Y/\|Y\|_F)$ maps the set $\mathcal{I}$ of full measure in $\lambda$ to the set $\pi(\mathcal{I})$ of full measure in $\sigma_k\times \sigma_l$. Indeed, since $\pi^{-1}(\pi(\mathcal{I}))=\bigcup_{\nu>0} \nu \mathcal{I}=\mathcal{I}$ and $(\lambda\times\lambda)(\mathcal{I}^c)=0$, we find  
\begin{gather*}
(\sigma_k\times\sigma_l)(\pi(\mathcal{I}))=\frac{(\lambda\times\lambda)\left(\pi^{-1}(\pi(\mathcal{I})^c)\cap( B_{m,k}\times B_{m,l})\right)}{\lambda\left(B_{m,k})\lambda(B_{m,l}\right)}\\
=\frac{(\lambda\times\lambda)\left(\pi^{-1}(\pi(\mathcal{I}))^c\cap (B_{m,k}\times B_{m,l})\right)}{\lambda\left(B_{m,k})\lambda(B_{m,l}\right)}= \frac{(\lambda\times\lambda)\left(\mathcal{I}^c\cap (B_{m,k}\times B_{m,l})\right)}{\lambda\left(B_{m,k})\lambda( B_{m,l}\right)}=0.
\end{gather*}

Let $S_k(m,k):=\{X\in S(m,k):\text{rank}\,X=k\}$ and consider the continuous mapping 
\begin{align}
\begin{split}
\varphi_k:S_k(m,k)&\to G(m,k)\\
   X&\mapsto \Pi_{\text{ran}\,X}.
   \end{split}
\end{align}
where $\Pi_{\text{ran}\,X}$ denotes the orthogonal projection on $\text{ran}\,X$. Furthermore let $\varphi:S_k(m,k)\times S_l(m,l)\to G(m,k)\times G(m,l),\ (X,Y)\mapsto (\varphi_k(X),\varphi_l(Y))$. Observe that $\varphi(\pi(\mathcal{I}))$ is precisely the set of $(P,P^\prime)\in G(m,k)\times G(m,l)$ such that $C|_{\text{ran}\,P\times \text{ran}\,P^\prime}$ is injective modulo scaling. We define a measure $\tilde{\mu}_k$ on $G(m,k)$ by setting $\tilde{\mu}_k(A):=\sigma_k(\varphi_k^{-1}(A))$ for all Borel subsets $A\in \mathcal{B}(G(m,k))$. From this one can seen that in the measure $\tilde{\mu}_k\times\tilde{\mu}_l$, the set $\varphi(\pi(\mathcal{I}))$ \footnote{Note that $\varphi^{-1}(\varphi(\pi(\mathcal{I})))=\pi(\mathcal{I})$.} has full measure and hence the following proposition concludes the proof of Theorem \ref{thm3}.
\begin{proposition}
The measure $\tilde{\mu}_k$ coincides with the the Haar measure $\mu_k$ on $G(m,k)$.
\end{proposition}
\begin{proof}
First note that $\tilde{\mu}_k(G(m,k))=\sigma_k(\varphi^{-1}(G(m,k)))=\sigma_k(S_k(m,k))=1$. Hence it suffices to check that $\tilde{\mu}(UAU^*)=\tilde{\mu}(A)$ for all $U\in U(m)$ and $A\in\mathcal{B}(G(m,k))$. Let   $U\in U(m)$ and $A\in\mathcal{B}(G(m,k))$. Then, since $U(m)S_k(m,k)=S_k(m,k)$,
\begin{align*}
\tilde{\mu}_k(UAU^*)&=\sigma_k(\{X\in S_k(m,k):\Pi_{\text{ran}\,X}\in UAU^*\})\\
				&=\sigma_k(\{X\in S_k(m,k):U^*\Pi_{\text{ran}\,X}U\in A\})\\
				&=\sigma_k(\{X\in S_k(m,k):\Pi_{\text{ran}\,U^* X}\in A\})\\
				&=\sigma_k(U\{X\in S_k(m,k):\Pi_{\text{ran}\,X}\in A\})\\
				&=\sigma_k(\{X\in S_k(m,k):\Pi_{\text{ran}\,X}\in A\}),
\end{align*}
where we used the unitary invariance of $\sigma$ in the last step.
\end{proof}
\section{Weak Identifiability Conditions for Deconvolution Maps}\label{app}
In this appendix we apply the techniques used in the present paper to the weak identifiability problem, achieving a small improvement with respect to Theorem 3.1 in \cite{li2015identifiability} (which agrees with our result except that a strict inequality $m>s_1+s_2$ is required) and showing optimality of the resulting bound. As the proofs are in large parts very similar to those of our main results, we omit some details.

%\begin{definition}(Definition 2.1 of \cite{li2015identifiability})
%Let $V\subseteq \mathbb{C}^{n_1}\setminus\{0\}$, $W\subseteq \mathbb{C}^{n_2}\setminus\{0\}$ be subsets. A pair of vectors $(u,v)\in V\times W$ is identifiable up to scaling with respect to a bilinear map $B\in\mathcal{B}(m)$ if $B(u,v)=B(u^\prime,v^\prime)$ for some other pair $(u^\prime,v^\prime)\in V\times W$ implies $[(u,v)]= [(u^\prime,v^\prime)]$.
%\end{definition}

\begin{theorem}[Weak identifiability conditions]
For $E\in F(m,k)$, let $\text{ran}(E)_s=\{x\in\text{ran}\,E:x\text{ is }s\text{-sparse when expanded in }E .\}$. Let $s_1,s_2\in\mathbb{N}_+$ be such that $s_1+s_2\leq m$ and let $k,l\in\mathbb{N}$ be such that $s_1\leq k\leq m$, $s_2\leq l\leq m$. Then, for Lebesgue almost all pairs $(E,D)\in F(m,k)\times F(m,l) $, the pair of vectors $(v,w)\in \text{ran}(E)_{s_1}\times \text{ran}(D)_{s_2}$ is identifiable up to scaling with respect to the circular convolution map $C:\mathbb{C}^m\times\mathbb{C}^m\to\mathbb{C}^m,\ (v,w)\mapsto v\circledast w$.
\end{theorem}
%This result can also be proved with the techniques discussed in this paper. For completeness, we include a proof sketch.
\begin{proof}
The proof of this theorem can be given along the lines of the proof of Theorem \ref{thm4} respectively Proposition \ref{propk}. The only difference is that the set $\mathcal{W}$ in Proposition \ref{propk} has to be replaced by the set $(vw^t-M^1_{s_1,s_2}(n_1,n_2))\setminus\{0\}$. Clearly $vw^t-M^1_{s_1,s_2}(n_1,n_2)$ is isomorphic to $M^1_{s_1,s_2}(n_1,n_2)$ and thus $\dim (vw^t-M^1_{s_1,s_2}(n_1,n_2))=\dim M^1_{s_1,s_2}(n_1,n_2)=s_1+s_2-1$ by the proof of Proposition \ref{prop1}.
\end{proof}

The following theorem shows that this bound is indeed optimal.
\begin{theorem}\label{thmid}
If a bilinear map $B\in\mathcal{B}(m)$ is such that $(u,v)\in\mathbb{C}^{n_1}_{s_1}\setminus\{0\}\times\mathbb{C}^{n_2}_{s_1}\setminus\{0\}$ is weakly identifiable up to scaling with respect to $B$ then $m\geq s_1+s_2$.
\end{theorem}

Let us first give two propositions that allow us to prove this theorem.
\begin{proposition}\label{propid1}
If there exists a bilinear map $B\in\mathcal{B}(m)$ such that $(v,w)\in\mathbb{C}^{n_1}_{s_1}\setminus\{0\}\times\mathbb{C}^{n_2}_{s_1}\setminus\{0\}$  is weakly identifiable up to scaling then there exists a linear map $M:M(s_1,s_2)\to \mathbb{C}^m$ such that $P(e_1e_1^t-M^1(s_1,s_2))\cap P(\text{ker}\,M)=\emptyset$ \footnote{Here $e_1$ denotes the first basis vector of the standard orthonormal basis of $\mathbb{C}^n$.}.
\end{proposition}
\begin{proof}
We stick to the notation introduced in Subsection \ref{sub1}. Let $(v,w)\in\mathbb{C}^{n_1}_{s_1}\setminus\{0\}\times\mathbb{C}^{n_2}_{s_1}\setminus\{0\}$ and let $M_B:M(n_1,n_2)\to\mathbb{C}^m$ be the linear map induced by $B$. Clearly there exist $A\in \mathcal{A}(n_1,s_1),\ B\in \mathcal{A}(n_2,s_2)$ such that $vw^t\in W_{A,B}$. Consider the isomorphism $\eta|_{M^1(s_1,s_2)}:M^1(s_1,s_2)\to W_{A,B}$ defined in \eqref{eta}. Let $U_1,U_2\in M(s_1,s_2)$ be unitaries such that $\eta(U_1e_1e_1^tU_2)=vw^t$. Define a linear map $M:M(s_1,s_2)\to \mathbb{C}^m$ by setting $M(X)=M_B\circ \eta(U_1XU_2)$ for all $X\in M(n_1,n_2)$. Now, assume for a contradiction that there is an $X=e_1e_1^t-\tilde{v} \tilde{w}^t$ for some $\tilde{v}\in\mathbb{C}^{s_1}\setminus\{0\},\,\tilde{w}\in\mathbb{C}^{s_2}\setminus\{0\}$ with $P(X)\in P(e_1e_1^t-M^1(s_1,s_2))\cap P(\text{ker}\,M)$. Then we have $B(v,w)-B(I_AU_1\tilde{v},I_BU_2\tilde{w})=M_B(vw^t-\eta(U_1\tilde{v} \tilde{w}^tU_2))=M(X)=0$, the sought contradiction.
\end{proof}
\begin{proposition}\label{propid2}
The set $P^{-1}\left(P(e_1e_1^t-M^1(s_1,s_2))\right)$ is the common zero locus of a set of homogeneous polynomials and $\dim P(e_1e_1^t-M^1(s_1,s_2))=s_1+s_2-1$.
\end{proposition}
\begin{proof}
For $1\leq i<i^\prime\leq s_1$ and $1\leq j<j^\prime\leq s_2$ define the $2\times 2$ minors 
\begin{align*}
M_{ij,i^\prime j^\prime}:M(s_1,s_2)\to \mathbb{C},\ X\mapsto\det\begin{pmatrix}
X_{ij} & X_{ij^\prime}\\
X_{i^\prime j} & X_{i^\prime j^\prime}
\end{pmatrix}.
\end{align*}
Then $V_G=P^{-1}(P(e_1e_1^t-M^1(s_1,s_2)))$ is the common zero locus of the set of homogeneous polynomials 
$G:=\{M_{ij,i^\prime j^\prime}\}_{1\leq i<i^\prime\leq s_1,1\leq j<j^\prime\leq s_2, (i,j)\neq (1,1)}$. To determine the dimension of $P(e_1e_1^t-M^1(s_1,s_2))$ consider the injective morphism
\begin{align*}
\eta:\mathbb{C}\times (M^1(s_1,s_2)\cap \{X\in M(s_1,s_2):X_{22}\neq 0\})\to V_G,\ (\lambda,X)\to \lambda e_1e_1^*+X.
\end{align*}
Now note that $\dim (M^1(s_1,s_2)\cap \{X\in M(s_1,s_2):X_{22}\neq 0\})=\dim M^1(s_1,s_2)=s_1+s_2-1$ \footnote{$M^1(s_1,s_2)$ is irreducible by Example 12.1 of \cite{harris2013algebraic} and hence $M^1(s_1,s_2)\cap \{X\in M(s_1,s_2):X_{22}\neq 0\}$ is irreducible as a non-empty open subset of the irreducible set $M^1(s_1,s_2)$ (see Exercise 1.1.3 of \cite{hartshorne1977algebraic}). Finally we have $\dim (M^1(s_1,s_2)\cap \{X\in M(s_1,s_2):X_{22}\neq 0\})=\dim M^1(s_1,s_2)$ by Exercise 1.6 and Proposition 1.10 of \cite{hartshorne1977algebraic}.} and hence we find $\dim P(e_1e_1^t-M^1(s_1,s_2))=\dim \mathbb{C}+\dim M^1(s_1,s_2)-1=s_1+s_2-1$.
\end{proof}

Now we are in a position to proof Theorem \ref{thmid}.
\begin{proof}[Proof of Theorem \ref{thmid}.]
Using Proposition \ref{propid1}  it suffices to show that if for a linear map $M:M(s_1,s_2)\to\mathbb{C}^m$ one has $P(\text{ker}\,M)\cap P(e_1e_1^t-M^1(s_1,s_2))=\emptyset$, then $m\ge s_1+s_2$. 

Clearly, for such an $M\in\mathcal{L}(m)$, $P(\text{ker}\,M)$ is a projective algebraic set and, by Proposition \ref{propid2}, $P(e_1e_1^t-M^1(s_1,s_2))$ is also a projective algebraic set. Then, the result follows again from the intersection theorem for projective varieties (cf. proof of Theorem \ref{thm1}).
\end{proof}

\end{document}